\documentclass[11pt]{article}
\usepackage{fullpage}

\usepackage{graphics}
\usepackage[dvips]{epsfig}

\usepackage{amsmath}
\usepackage{amssymb}
\usepackage{amsfonts}
\usepackage{graphicx}

\usepackage{algorithm} 
\usepackage[noend]{algorithmic}
\usepackage{amsthm}
\usepackage{color}

\newtheorem{theorem}{Theorem}[section]
\newtheorem{lemma}{Lemma}[section]

\newtheorem{proposition}{Proposition}[section]

\newcommand{\remove}[1]{}

\begin{document}

\baselineskip  0.2in 
\parskip     0.1in 
\parindent   0.0in 

\title{{\bf Asynchronous Broadcasting with Bivalent Beeps}\footnote{A preliminary version of this paper appeared in the Proc. 23rd International Colloquium
on Structural Information and Communication Complexity (SIROCCO 2016), LNCS 9988.}}

\author{
Kokouvi Hounkanli\thanks{D\'{e}partement d'informatique, Universit\'{e} du Qu\'{e}bec en Outaouais,
Gatineau, Qu\'{e}bec J8X 3X7,
Canada. {\tt houk06@uqo.ca}}
\and
Andrzej Pelc\thanks{D\'{e}partement d'informatique, Universit\'{e} du Qu\'{e}bec en Outaouais,
Gatineau, Qu\'{e}bec J8X 3X7,
Canada. {\tt  pelc@uqo.ca}.
Supported in part by NSERC discovery grant 8136 -- 2013
and by the Research Chair in Distributed Computing of
the Universit\'{e} du Qu\'{e}bec en Outaouais.}
}

\date{ }

\maketitle

\begin{abstract}
 In broadcasting, one node of a network has a message that must be learned by all other nodes. We study deterministic algorithms for this fundamental communication task in a very weak model
 of wireless communication.
 The only signals sent by nodes are {\em beeps}. Moreover, they are delivered to neighbors of the beeping node in an asynchronous way: the time between sending and reception
 is finite but unpredictable. We first observe that under this scenario, no communication is possible, if beeps are all of the same strength. Hence we study broadcasting in the
 {\em bivalent beeping model}, where every beep can be either {\em soft} or {\em loud}.  At the receiving end, if exactly one soft beep is received by a node in a round, 
  it is heard as soft. Any other combination of beeps received in a round is heard as a loud beep. The cost of a broadcasting algorithm is the total number of beeps sent by all nodes.
  
 We consider four levels of knowledge that nodes may have about the network: anonymity (no knowledge whatsoever), ad-hoc (all nodes have distinct labels and every node knows only its own label),
 neighborhood awareness (every node knows its label and labels of all neighbors), and full knowledge (every node knows the entire labeled map of the network and the identity of the source).
 We first show that in the anonymous case, broadcasting is impossible even for very simple networks. For each of the other three knowledge levels we provide upper and lower bounds on the minimum cost of 
 a broadcasting algorithm. Our results show separations between all these scenarios. Perhaps surprisingly, the jump in broadcasting cost between the ad-hoc and neighborhood awareness levels is much larger than between the neighborhood awareness and full knowledge levels, although in the two former levels knowledge
 of nodes is local, and in the latter it is global.

\vspace{2ex}


\end{abstract}


\vfill

\thispagestyle{empty}

\pagebreak

\section{Introduction}
{\bf The background and the problem.}
Broadcasting is a fundamental communication task in networks. One node of a network, called the {\em source}, has a message that must be learned by all other nodes.  
We study deterministic algorithms for this well-researched task in a very weak model of wireless communication.
The only signals sent by nodes are {\em beeps}. Moreover, they are delivered to neighbors of the beeping node in an asynchronous way: the time between sending and reception
 is finite but unpredictable. Our aim is to study how the combination of two weaknesses of the communication model, very simple and short messages on the one hand,  and the asynchronous way of delivery
 on the other hand, influences efficiency of communication. Each of these two model weaknesses 
 separately has been studied before. Synchronous broadcasting and gossiping with beeps was studied in \cite{CD}. Asynchronous broadcasting in the radio model, where large messages can be sent in 
 a round, was investigated in \cite{CFP,CR,Pe2}. To the best of our knowledge, the combination of the two model weaknesses, i.e., very short messages and asynchronous delivery, has never been studied before.

We first observe that under this very harsh scenario, no communication is possible, if beeps are all of the same strength (see Section 2). Hence we study broadcasting in the
 {\em asynchronous bivalent beeping model}, where every beep can be either {\em soft} or {\em loud}, as this is, arguably, the weakest model under which asynchronous wireless broadcasting can be performed.  
 At the receiving end, if exactly one soft beep is received by a node in a round, 
  it is heard as soft. Any other combination of beeps received in a round is heard as a loud beep. The cost of a broadcasting algorithm is the total number of beeps sent by all nodes. This measures
  (the order of magnitude of) the energy consumption by the network, as the energy used to send a loud beep can be considered to be a constant multiple of that used to send a soft beep.
  
  {\bf The model.}
Communication proceeds in rounds. In each round, a node can either listen, i.e., stay silent,
or send a  {\em soft beep}, or send a {\em loud beep}. For any beep sent by any node, an omniscient asynchronous adversary chooses a non-negative integer $t$, and delivers it to all neighbors of the sending node $t$ rounds later. The delivery delay at all neighbors is the same for a given beep, but may be different for different beeps of the same node and for beeps of different nodes. The only rule that the adversary has to obey regarding delivery of different beeps sent by the same node, is that they must be delivered in the same order as they were sent, and cannot be collapsed in delivery, i.e. two beeps cannot be delivered as one beep. This type of asynchronous adversary was called the {\em node adversary} in \cite{CR} and the {\em strong adversary} in \cite{CFP}. The motivation is similar as in \cite{CR,CFP}. Nodes execute the broadcasting protocol concurrently with other tasks. Beeps to be sent by a node are prepared for transmission (stored), and then each beep (soft or loud) is transmitted in order. The
(unknown) delay between these actions is
decided by the adversary. In our terminology,  storing for transmission corresponds to sending and
actual transmission corresponds to simultaneous delivery to all neighbors.  We assume that, at
short distances between nodes, the travel time of the beep is negligible. The
delay between storing and transmitting (in our terminology, between sending
and delivery) depends on how busy the node is with other concurrently performed computational
tasks. 

At the receiving end, a node can hear something only if it is silent in the delivery round.  If exactly one soft beep is delivered to a node in a round, 
it is heard as soft. Any other combination of beeps delivered to a node from its neighbors in a round (a single strong beep, or more than one beep of any kind) is heard as a loud beep. This way of modeling reception corresponds to a threshold in the listening device: the strength of a single soft beep is below the threshold, and
the strength of a loud beep,
or the combined strength of more than one beep is above the threshold.
The cost of a broadcasting algorithm is the total number of beeps sent by all nodes.

The network is modeled as an $n$-node simple connected undirected graph, referred to as {\em graph}. We use terms ``network'' and ``graph'' interchangeably. 
We consider four levels of knowledge that nodes may have about the network:
\begin{enumerate}
\item
anonymous networks:  nodes do not have any labels and know nothing about the network;
\item
ad-hoc networks: all nodes have distinct labels and every node knows only its own label;
\item
neighborhood-aware networks: all nodes have distinct labels, and every node knows its label and labels of all neighbors;
\item
full-knowledge networks: all nodes have distinct labels, every node knows the entire labeled map of the network and the identity of the source.
\end{enumerate}

The messages to be broadcast are from some set of size $M$, called the 
{\em message space}. Without loss of generality, let the message space be the set of integers
$\{0,\dots , M-1\}$.
Except for the anonymous networks, all nodes have different labels from the set of integers $\{0,\dots , L-1\}$, called the {\em label space}.

{\bf Our results.} 
Our aim is to study how different levels of knowledge about the network influence feasibility and cost of broadcasting in the asynchronous bivalent beeping model.
We first show that, in the anonymous case, broadcasting is impossible even for very simple networks. For each of the other three knowledge levels, broadcasting is feasible, and we provide upper and lower bounds on the minimum cost of a broadcasting algorithm, in terms of the sizes of the network, of the message space and of the label space. Showing an upper bound $UB$ on the cost of broadcasting at a given knowledge level means showing an algorithm which accomplishes broadcasting at this cost, for any network with this knowledge level, and any message to be broadcast. Showing a lower bound $LB$ means that, for any algorithm of lower cost, there is some network at this knowledge level, and some message for which the algorithm fails.

For ad-hoc networks we give an algorithm of cost $2^{O(L+M)^2}$ \footnote{If one of the parameters, $L$ or $M$, is known to the nodes, this complexity can be decreased to $2^{O(LM)}$ (see section 4).}. Since this cost is very large, it is natural to ask if there are broadcasting algorithms
of cost polynomial in $L$ and $M$. The answer turns out to be negative: indeed, we prove a lower bound of $\Omega(2^L)$ on the cost  of any broadcasting algorithm in ad-hoc networks.
For neighborhood-aware networks we prove an upper bound  of $O(n\log M+ e\log L)$, where $n$ is the number of nodes and $e$ is the number of edges, and a lower bound of $\Omega(n\log M +n \log\log L)$.
Finally, for full-knowledge networks, we provide matching upper and lower bounds of $\Theta(n\log M)$.

Note that the above bounds show separations, in terms of broadcasting cost, between all the knowledge levels, in the case often appearing in applications, when the message space is some predetermined
dictionary independent of the network, i.e., its size $M$ is $O(1)$. Indeed, since $L \geq n$, the lower bound $\Omega (2^L)$ for ad-hoc networks exceed the (worst-case) upper bound $O(n^2\log L)$ for known-neighborhood networks,
and the lower bound $\Omega(n\log\log L)$ for  known-neighborhood networks exceed the tight bound $\Theta(n)$ for full-knowledge networks. 

It is interesting to compare the sizes of the two broadcasting cost jumps: the jump between ad-hoc and known-neighborhood networks, and the jump between known-neighborhood and full-knowledge networks. We illustrate it for the commonly assumed case, when the size $L$ of the label space is polynomial in the size $n$ of the network (and the size of the message space is $O(1)$, as before).
The first jump
is at least from $\Omega (2^n)$ to $O(n^2\log n)$, i.e.,  exponential in $n$. The second jump is at most from $O(n^2\log n)$  to $\Theta(n)$, i.e., polynomial in $n$.
This may seem slightly counterintuitive, because both in ad-hoc and in known-neighborhood networks, information available to nodes is local, while in full-knowledge networks it is global.
So at first glance it would seem that the larger jump should occur between known-neighborhood and full-knowledge networks.


{\bf Related work.}
Broadcasting has been studied in various models for over four decades. Early work focused on the telephone model, where in each round communication
proceeds between pairs of nodes forming a matching. Deterministic broadcasting in this model has been studied, e.g.,  in \cite{SCH}. In \cite{FPRU} the authors studied randomized broadcasting.
In the telephone model, studies focused on the time of the communication task and on the number of messages it uses. Early literature on communication in the telephone and related models is surveyed in \cite{FrLa,HHL}. In \cite{AGPV} the authors studied tradeoffs between the radius within which nodes know the network and broadcasting efficiency in the message passing model. Fault-tolerant aspects of broadcasting and gossiping are surveyed in~\cite{Pe}. 

More recently, broadcasting has been studied in the radio model. While radio networks are used to model wireless communication, similarly as the beeping model, 
in radio networks nodes send entire messages of some bounded, or even unbounded size in a single round, which makes communication drastically different from that in the beeping model.
The focus in the literature on radio networks was usually on the time of communication.
Deterministic broadcasting in the radio model was studied, e.g., in \cite{CGR,KP}, and randomized broadcasting was studied in
\cite{KM}. The book \cite{Ko} is devoted to algorithmic aspects of communication in radio networks.

In all the above papers, radio communication was supposed synchronous, i.e., the message was delivered in the same round in which it was sent. Asynchronous broadcasting in radio networks
was studied in \cite{CFP,CR,Pe2}. It is important to stress a significant difference between the radio and the beeping models, in the context of asynchrony. Since in the radio model large messages can be sent and delivered in a single round, asynchrony cannot alter a message, it can only destroy it, by creating unwanted interference. In the beeping model, however, beeps from various senders can be simultaneously delivered by the adversary, thus altering the intended numbers and types of beeps, creating ``new'' messages. 

The beeping model has been introduced in \cite{CK} for vertex coloring, and used
in \cite{AABCHK} to solve the MIS problem, and in \cite{YJYLC} to construct a minimum connected dominating set.
Randomized leader election in the radio and in the beeping model was studied in \cite{GH}. Deterministic leader election in the beeping model was investigated in 
\cite{FSW}. In \cite{HMP}, the authors studied the tasks of global synchronization and consensus using beeps, in the presence of faults. 
In \cite{GN}, the authors studied the quantity of computational resources needed to solve problems in complete networks using beeps. In \cite{MRZ}, various distributed problems were
investigated under several variations of the beeping model from \cite{CK}, and randomized emulations between these models were shown.
In \cite{EP}, the authors studied the task of rendezvous of agents communicating by beeps. The time of synchronous broadcasting and gossiping with beeps was studied in \cite{CD}.

\section{Preliminaries}

The following observation shows that asynchronous broadcasting with beeps of uniform strength is impossible even in very simple graphs. This is the reason why we use the bivalent beeping model.

\begin{proposition}
Asynchronous broadcasting using beeps of uniform strength is impossible even in the two-node graph.
\end{proposition}

\begin{proof}
Consider two source messages, $m_1$ and $m_2$, that have to be transmitted from one node to the other in the two-node graph. Suppose that the source sends $k_1$ beeps for message $m_1$ and $k_2$ beeps for message $m_2$, where $k_1 \leq k_2$, without loss of generality. The adversary delivers the beeps for message $m_1$ in consecutive rounds $r,r+1,\dots,r+k_1-1$. Suppose that $s$ is the round in which the receiving node correctly decodes message $m_1$. Then, for message $m_2$, the adversary delivers $k_1$ beeps in rounds $r,r+1,\dots,r+k_1-1$, and the remaining $k_2-k_1$ beeps in rounds
$t+1,\dots, t+k_2-k_1$, where $t=\max(s,r+k_1-1)$. However, in round $s$, the receiving node has exactly the same information as for message $m_1$, and hence it incorrectly outputs the message as $m_1$. 
\end{proof}

In the rest of the paper we use the asynchronous bivalent beeping model, described in the introduction.

\section{Anonymous networks}

In this section we show that, if nodes do not have labels, then broadcasting is impossible, even for very simple graphs, and even when nodes know the topology of the network.

\begin{proposition}
Broadcasting for anonymous networks is impossible even in the cycle of size 4.
\end{proposition}

\begin{proof}
Consider the anonymous cycle of size 4, and consider a hypothetical broadcasting algorithm $A$. For convenience, we label nodes $a,b,c,d$, in clockwise order. This is for the negative argument
only: nodes do not have access to these labels. Suppose that node $a$ is the source. 
Notice that, in any execution of algorithm $A$, nodes $b$ and $d$ send exactly the same beeps in the same rounds, as in each round they have the same history: indeed,
they receive the same beeps in the same rounds, they are identical, and execute the same deterministic algorithm. Let  $m_1$ and $m_2$ be two different messages that have to be broadcast by the source. 
Consider two executions of the algorithm $A$: execution  $E_1$, in which the source broadcasts message $m_1$, and execution  $E_2$, in which the source broadcasts message $m_2$.
Let $s_1$ be the sequence of beeps (soft or loud) sent by $b$ and $d$ in execution $E_1$ and let  $s_2$ be the sequence of beeps sent by $b$ and $d$ in execution $E_2$. Let $k_1$ be the length of $s_1$,
and let $k_2$ be the length of $s_2$,  where $k_1 \leq k_2$, without loss of generality. In both executions, the adversary delivers consecutive beeps from $b$ and from $d$ in the same rounds.
As a result, node $c$ hears only loud beeps: $k_1$ of them in execution $E_1$, and $k_2$ of them in execution $E_2$. The choice of the rounds of delivery of bits from $b$ and $d$ is as follows.
In execution $E_1$ these are consecutive rounds $r,r+1,\dots,r+k_1-1$, starting from some round $r$. Suppose that $s$ is the round in which  node $c$ correctly outputs message $m_1$.
Then, in execution $E_2$, the adversary delivers the first $k_1$ beeps from $b$ and $d$ in rounds $r,r+1,\dots,r+k_1-1$, and the remaining $k_2-k_1$ beeps in rounds
$t+1,\dots, t+k_2-k_1$, where $t=\max(s,,r+k_1-1)$. In round $s$, node $c$ has the same history in executions $E_1$ and $E_2$: it heard a loud beep in the same rounds, in both these executions.
Hence, in execution $E_2$, it incorrectly outputs the message  $m_1$ in round $s$.
\end{proof}

\section{Ad-hoc networks}

In this section we show that providing nodes with distinct labels makes broadcasting possible in arbitrary graphs, even if nodes do not have any initial knowledge about the network, except their own label.
Let $\cal{N}$ denote the set of non-negative integers. Consider the function $\varphi : \cal{N} \times \cal{N} \longrightarrow  \cal{N} $ given by the formula 
$\varphi(x,y)=x+(x+y)(x+y+1)/2$. This is a bijection with the property
$\varphi(x,y) \in O((x+y)^2)$. Intuitively, this is the ``snake function'' arranging all couples of non-negative integers into one infinite sequence. 

The following algorithm is executed by an active node with label $\ell$. In the beginning, all nodes are active. 
The part {\em Receive} is executed by any node other than the source. Its result is outputting the source message. This part is skipped by the source, as it knows the message. The part {\em Send} is executed by the source at the beginning of the algorithm, and it is executed by every other node upon outputting the source message in the part {\em Receive}.
After executing the part {\em Send}, the node becomes non-active.

\vspace{3ex}
\noindent
{\bf Algorithm} {\tt Ad-hoc}

\vspace{2ex}
\noindent
Part 1. {\em Receive}\\
Wait until the number of soft beeps received is at least 1/2 of the number of loud beeps received.\\
Let $t$ be the number of loud beeps received, and let $z$ be the largest integer such that $8^z \leq t$.\\
Compute the unique couple of non-negative integers $(x,y)$, such that  $\varphi(x,y)=z$.\\
Output $y$ as the source message.

\vspace{2ex}
\noindent
Part 2. {\em Send}\\
Compute $\varphi(\ell,y)$, where $y$ is the source message.\\
Send $8^{\varphi(\ell,y)}$ loud beeps, followed by $8^{\varphi(\ell,y)}$  soft beeps. 
\hfill$\diamond$

\vspace{2ex}

The following result shows that Algorithm {\tt Ad-hoc} is correct, and estimates its cost.

\begin{theorem}
Upon completion of Algorithm {\tt Ad-hoc} in an arbitrary graph, every node correctly outputs the source message.
The cost of the algorithm is $2^{O((L+M)^2)}$.
\end{theorem}

\begin{proof}
The proof of correctness is split into two parts. We first show that no node outputs the source message incorrectly, and then we prove that every node outputs the source message in finite time.
Let $m$ be the source message. The first part of the proof is by contradiction. Suppose that some node outputs the source message incorrectly,  let $r$ be the first round when this happens, and let $u$
be a node with label $\ell$, incorrectly outputting the source message in round $r$. 
Let $u_1, \dots, u_k$ be the nodes adjacent to $u$ whose at least one beep is delivered by round $r$, ordered in increasing order of their labels $\ell_1,\dots,\ell_k$.
Since $u$  outputs the source message in round $r$, the set of nodes $\{u_1, \dots, u_k\}$ is non-empty. Moreover, all nodes $u_1, \dots, u_k$ must have outputted the source message before round $r$ (because they already sent some beeps by round $r$), and hence they outputted it correctly. Let $1\leq i \leq k$ be the largest integer $j$, such that at least one soft beep of node $u_j$ was delivered by round $r$. 

Suppose that $t$ was the number of loud beeps heard by $u$ by the round $r$. Since $u$ outputted the source message incorrectly, the largest integer $z'$,
such that $8^{z'} \leq t$, cannot be equal to $z=\varphi(\ell_i,m)$. (If it were, node $u$ would correctly compute the source message $m$ because $\varphi$ is a bijection.)
The integer $z'$ cannot be smaller than $z$ because node $u$ heard at least $8^z$ loud beeps sent by node $u_i$. Hence $z'\geq z+1$. This implies that node $u$ must have heard at least $8^{z+1}$
loud beeps by round $r$. How many soft beeps could it hear by round $r$? All these beeps could come only from nodes $u_1,\dots,u_i$. The total number of soft beeps sent by these nodes is
$\sum_{j=1}^i 8^{\varphi(\ell_j,m)}$. Since $\varphi(\ell_1,m)<\varphi(\ell_2,m)< \cdots <\varphi(\ell_i,m)$, we have $\sum_{j=1}^i 8^{\varphi(\ell_j,m)}<\frac{8}{7} \cdot 8^{\varphi(\ell_i,m)}=\frac{8}{7}\cdot 8^z$.
On the other hand, the number of soft beeps heard by node $u$ by round $r$ must be at least 1/2 of the number of loud beeps it heard by round $r$. This implies $\frac{8}{7}\cdot 8^z \geq \frac{1}{2}\cdot 8^{z+1}$,
which is a contradiction. This completes the first part of the proof.

We now prove that every node outputs the source message in finite time. This part of the proof is also by contradiction. Suppose that some node never outputs the source message. 
Since the source itself knows the source message, and the graph is connected, there must exist adjacent nodes $u$ and $v$, such that $u$ outputs the source message in finite time and 
$v$ does not. Let $v_1, \dots, v_s$ be the nodes adjacent to $v$ that ever send at least one beep, ordered in increasing order of their labels $\lambda_1,\dots,\lambda_s$.
The set of nodes $\{v_1, \dots, v_s\}$ is non-empty.
We show that, at some point, the number of soft beeps heard by $v$ is at least 1/2 of the number of loud beeps heard by $v$. 
Indeed, assume that this did not happen before all beeps of all nodes
$v_1, \dots, v_s$ are delivered. The number of all beeps sent by nodes $v_1, \dots, v_{s-1}$ is  $2\cdot \sum_{j=1}^{s-1} 8^{\varphi(\lambda_j,m)}$.
Since $\varphi(\lambda_1,m)<\varphi(\lambda_2,m)< \cdots <\varphi(\lambda_{s},m)$, we have $2\cdot \sum_{j=1}^{s-1} 8^{\varphi(\lambda_j,m)}<\frac{2}{7}\cdot 8^{\varphi(\lambda_s,m)}$.
In the worst case, these beeps can be delivered by the adversary simultaneously with the same number (fewer than $\frac{2}{7}\cdot 8^{\varphi(\lambda_s,m)}$) of soft beeps sent by node $v_s$, thus producing loud beeps heard by node $v$. This would decrease the number of soft beeps heard by $v$ and increase the number of loud beeps heard by this node, but the change cannot be too big. Indeed, 
this gives fewer than $\frac{9}{7}\cdot 8^{\varphi(\lambda_s,m)}$ loud beeps heard by $v$. On the other hand, node $v$ hears at least  $\frac{5}{7}\cdot 4^{\varphi(\lambda_s,m)}$ soft beeps sent by $v_s$ and left intact (not delivered simultaneously with other beeps) by the adversary. Hence the number of soft beeps heard by node $v$ is at least 1/2 of the number of loud beeps heard by it. It follows that node $v$ outputs the source message, contrary to our assumption.

This completes the proof of correctness of Algorithm {\tt Ad-hoc}. We now estimate its cost. A node with label $\ell$ sends $2\cdot 8^{\varphi(\ell,m)}$ beeps, where $m$ is the source message.
Hence the cost of Algorithm {\tt Ad-hoc} in an $n$-node network is at most $2n\cdot 8^{\varphi(L,M)}$. Since $\varphi(L,M) \in O((L+M)^2)$, and $\varphi(L,M) \geq L \geq n$, this gives the cost
$2^{O((L+M)^2)}$.
\end{proof}

%
%
%

{\bf Remark.} Notice that, if nodes know one of the parameters, either $L$ or $M$, then the bijection $\varphi$ can be replaced by a more efficient one-to-one function from 
the product $\{0,\dots , L-1\}\times \{0,\dots , M-1\}$ to non-negative integers. For example, if $L$ is known, then this function can be $\psi(\ell,m)=mL+\ell$, and if 
$M$ is known, then this function can be $\psi '(\ell,m)=\ell M+m$. The values of these functions are in $O(LM)$, and hence, if we substitute one of them for $\varphi$,
the cost of the algorithm becomes $2^{O(LM)}$.

As we have seen above, the cost of Algorithm {\tt Ad-hoc} is very large: even with knowledge of $L$ or $M$, it is exponential in the product of these parameters.
Hence, it is natural to ask if there are broadcasting algorithms, for ad-hoc networks,
with cost polynomial in $L$ and $M$. Our next result shows that the answer is negative. Before proving it we recall a notion and a fact from  \cite{CFP}.

A set $S$ of positive integers is \emph{dominated} if, for any finite subset $T$
of
$S$, there exists $t\in T$ such that $t$ is larger than the sum of all
$t' \ne t $ in $T$.

\begin{lemma}\label{lemma:2^k}
Let $S$ be a finite dominated set and let $k$ be its size. Then there exists $x
\in S$
such that $x\ge 2^{k-1}$.
\end{lemma}

\begin{theorem}
For arbitrary integers $L \geq 4$, there exist $L$-node ad-hoc networks, for which
the cost of every broadcasting algorithm is $\Omega(2^L)$.
\end{theorem}

\begin{proof}

Let $A$ be any broadcasting algorithm. 
For any set $S\subseteq \{1,\dots, L-2\}$, of size at least 2, the graph $G_S$ is defined as follows.
$G_S$ has $|S|+2$ nodes with labels from the set $S \cup \{0,L-1\}$.
Each of the nodes with labels in $S$ is adjacent to each of the nodes with labels $0$ and $L-1$, 
and there are no other edges in the graph. The node with label $0$ is the source, and the node with label $L-1$ is called the {\em sink}.


We will consider executions of algorithm $A$ in graphs $G_S$, in which the adversary obeys the following rules concerning the delivery of beeps sent by the source and the sink:
\begin{enumerate}
\item
All beeps sent by the source after it heard some beep, are delivered after the round when the sink outputs the source message.
\item
All beeps sent by the sink are delivered after the round in which the sink outputs the source message.
\end{enumerate}

Since the considered networks are ad-hoc, i.e., a priori, every node knows only its own label, and the adversary obeys the above rules, the number of beeps sent  by a node with a given label $\ell \in \{1,\dots, L-2\}$ by the round in which the sink outputs the source message, depends only on this label and on the source message, and not on the graph $G_S$ in which the algorithm is executed. Indeed, the history of a node with label $\ell \in \{1,\dots, L-2\}$, by the round in which the sink
outputs the source message, is the same in all graphs $G_S$, for a given source message $m$.

Consider the execution of algorithm $A$ in the graph $G_{\{1,\dots, L-2\}}$, for a fixed source message $m$. Let $B(\ell)$, for $\ell \in \{1,\dots, L-2\}$,
be the number of beeps of both kinds, that the node with label $\ell$ sends by the round in which the sink outputs the source message. If the set  of integers $I=\{B(\ell): \ell \in \{1,\dots, L-2\}\}$ is dominated, then by
Lemma \ref{lemma:2^k}, some integer in this set is at least $2^{L-3}$, and we are done. Otherwise, there exists a subset $T\subseteq  \{1,\dots, L-2\}$,
with the following property. If $t\in T$ is such that $B(t) \geq B(t')$, for all $t' \in T \setminus \{t\}$, then $B(t)\leq \sum_{t' \in T \setminus \{t\}}B(t')$.  
Consider the execution $E$ of algorithm $A$ in the graph $G_T$, for the same source message $m$. As observed above, the number of beeps
of both kinds, that the node with label $\ell$ sends in this execution by the round in which the sink outputs the source message, 
is $B(\ell)$. The adversary delivers beeps  sent by nodes with labels from $T$, in consecutive rounds,
delivering simultaneously a beep sent by the node with label $t$ with one or more beeps sent by nodes with labels $t' \in T \setminus \{t\}$, in such a way that
in no round a single beep is delivered. This is possible due to the inequality $B(t)\leq \sum_{t' \in T \setminus \{t\}}B(t')$. Hence the sink hears only
loud beeps.

Now, consider a different source message $m'$. The same argument as above shows that,  if  the cost of the algorithm $A$ on the graph $G_{\{1,\dots, L-2\}}$
is smaller than  $2^{L-3}$, then there exists some set $T'\subseteq  \{1,\dots, L-2\}$, such that, in the execution $E'$ of the algorithm $A$ on the graph $G_{T'}$,
with the source message $m'$, the sink hears only loud beeps.

Suppose that, by the time it outputs the source message, the sink hears $k$ loud beeps in the execution $E$ and hears $k'$ loud beeps in the execution $E'$.
Without loss of generality, assume that $k \leq k'$. The choice of rounds of delivery  of these beeps by the adversary is the following.

In execution $E$, these are consecutive rounds $r,r+1,\dots,r+k-1$, starting from some round $r$. Suppose that $s$ is the round in which the sink correctly outputs message $m$.
Then, in execution $E'$, the adversary first delivers beeps in rounds $r,r+1,\dots,r+k-1$, and the remaining $k'-k$  rounds of beep delivery are
$z+1,\dots, z+k'-k$, where $z=\max(s,,r+k-1)$. In round $s$, the sink has the same history in executions $E$ and $E'$: it heard only loud beeps, and this happened in the same rounds, in both these executions.
Hence, in execution $E'$, it incorrectly outputs the message $m$ in round $s$.

The obtained contradiction comes from assuming that the cost of algorithm $A$ on the graph $G_{\{1,\dots, L-2\}}$ is smaller than $2^{L-3}$, for all source messages. 
This completes the proof.
\end{proof}

\section{Neighborhood-aware networks}

In this section we assume that all  nodes have distinct labels, and that each of them knows its own label and the labels of all its neighbors. This seemingly small increase of knowledge,
compared to ad-hoc networks
(the knowledge of every node is still local) turns out to decrease the cost of broadcasting in a dramatic way. In order to guarantee a low cost of broadcasting, 
we have to encode messages by sequences of beeps very efficiently. The algorithm uses messages of two types: non-negative integers and triples of non-negative integers. These messages have to be encoded by strings of beeps of length logarithmic
in the values of these integers, in such a way that the recipient knows when the string starts and  ends, and can unambiguously decode the message from the string. However, as opposed to Algorithm {\tt Ad-hoc} in which nodes sent exponentially many beeps, such efficient encoding is very vulnerable to possible actions of the adversary that can  arbitrarily interleave delivered beeps coming from different neighbors of a node. In order to avoid this, we design our algorithm in such a way that beeps encoding a message sent by some node are delivered before any other node starts sending its own beeps. In this way, the danger of interleaving beeps is avoided. 

Before presenting the algorithm, we define the encoding of integers and of their triples, announced above. We denote a loud beep by $l$, a soft beep by $s$, and we use the symbol $\cdot$ for the concatenation of sequences of beeps. Let $k$ be a non-negative integer, and let $(c_1,\dots,c_r)$ be its binary representation. Denote by $S(k)$
the sequence of $2r$ beeps resulting from $(c_1,\dots,c_r)$ by replacing every bit $c_i=0$ by $(ls)$ and by replacing every bit $c_j=1$ by $(sl)$. The code of an integer
$k$, denoted by $[k]$, is the sequence $(ll)\cdot S(k) \cdot (ll)$. The code of a triple $(a,b,c)$ of integers, denoted by $[a,b,c]$, is the sequence
$(ll)\cdot S(a) \cdot (ss)\cdot S(b)\cdot (ss)\cdot S(c) \cdot (ll)$. Note that a sequence of 2 loud beeps marks the beginning and end of a message, and 
all messages contain an even number of beeps, logarithmic in the integers transmitted. A node at the receiving end can determine the beginning of the message 
as a sequence $\sigma$ of 2 consecutive loud beeps, and the end of the message as the first sequence $\sigma'$ of 2 consecutive loud beeps starting after the end of $\sigma$ at an odd position, where the first bit of the sequence $\sigma$ is at position 1. In order to decode the content of the message $(ll)\cdot \alpha \cdot (ll)$, with the beginning and end already correctly identified, a node looks for separators $(ss)$ starting at odd positions of $\alpha$. There are either 0 or 2 such separators. In the first case, the transmitted message was an integer, and the node decodes its binary representation by replacing each couple $(ls)$ by 0 and each couple $(sl)$ by 1. In the second case, the node
can unambiguously represent $\alpha$ as $\alpha_1 \cdot (ss)\cdot \alpha_2 \cdot (ss) \cdot \alpha _3$, where each $\alpha_i$ has even length, and decode $\alpha_1,\alpha_2, \alpha_3$ as above. 

Using the above encoding, we are now able to describe our broadcasting algorithm.
At a high level, it is organized as a depth-first traversal of the graph, starting from the source. 
We will use the instructions ``send $[a]$'' and ``send $[a,b,c]$''  that are procedures sending the above described sequences of beeps, in consecutive rounds.
A message $[a]$, where $a \in \{0,1,\dots M-1\}$, is always the source message to be broadcast. There are two kinds of messages of type ``triple of integers'':
For $a,b \in \{0,1,\dots L-1\}$, a  message of the form $[a,b,0]$ corresponds to a forward DFS edge traversal from the node with label $a$ to a node with label $b$,
and  a  message of the form $[a,b,1]$ corresponds to a backward DFS edge traversal from the node with label $a$ to a node with label $b$.
 
 The algorithm is executed by a node with label $\ell$. The actions of the node alternate between executing ``send'' instructions and listening. The algorithm is organized in such a way that
 the following {\em disjointness property} is satisfied.  Consider a node $u$ executing some send instruction $I(u)$. Let $\sigma(u)$ be the segment of consecutive rounds between the sending of the first
 beep of instruction $I(u)$ and the delivery of the last bit of this instruction. Then, for any two nodes $u$ and $v$, executing any send instructions $I(u)$ and $I(v)$, the segments of rounds $\sigma(u)$ and $\sigma(v)$
 are disjoint. This property permits to identify circulating messages as distinct ``packets'', and use them to implement a DFS traversal.
  
 When the node listens, it watches for the beginning and end of a message formed by the delivered beeps. When it detects a complete message, it reacts to it in one of two ways: it either keeps listening and watches for another complete message, or it reacts by executing some ``send'' instruction. More specifically, the actions of the node with label $\ell$,
 other than the source, are as follows. 
 After getting the source message and the first forward DFS message $[a,\ell,0]$, addressed to it and coming from a node with label $a$, the node with label $\ell$ starts spreading the message to all its neighbors with labels $a_i$,
 except that with label $a$, by sending the decoded source message $[m]$ and sending forward DFS messages $[\ell,a_i,0]$ addressed to them, in increasing order of labels. In order
 to transit from one neighbor to the next, the node $\ell$ waits for a backward message $[a_i,\ell,1]$, addressed to it. In the meantime, node $\ell$ refuses all subsequent
 forward DFS messages $[b,\ell,0]$ , for $b\neq a$, addressed to it, responding by a backward DFS message $[\ell,b,1]$. The actions of the source are similar. 
 
 The pseudocode of the algorithm follows.
 

\vspace{3ex}
\noindent
{\bf Algorithm} {\tt Neighborhood-aware}

\vspace{2ex}
\noindent
{\bf if} the executing node is the source, and the source message is $m$  {\bf then}\\
\hspace*{1cm}$message \leftarrow m$\\
\hspace*{1cm}let $(a_1,a_2,\dots ,a_s)$ be labels of all the neighbors of the node,\\ 
\hspace*{1cm}in increasing order\\
\hspace*{1cm}{\tt Spread}$(a_1,\dots,a_s)$\\
\hspace*{1cm}whenever a message $[b,\ell,0]$, for some integer $b$, is decoded then\\
 \hspace*{2cm}send $[\ell,b,1]$\\
{\bf else}\\
\hspace*{1cm}when a message $[m]$ is decoded for the first time, then\\
\hspace*{2cm}$message \leftarrow m$\\ 
\hspace*{2cm}output $message$ as the source message\\
\hspace*{1cm}when a message $[a,\ell,0]$ is decoded for the first time, then\\
\hspace*{2cm}let $(a_1,a_2,\dots ,a_s)$ be labels of all the neighbors of the node,\\ 
\hspace*{2cm}except $a$, in increasing order\\
\hspace*{2cm}{\tt Spread}$(a_1,\dots,a_s)$\\
\hspace*{2cm}send $[\ell,a,1]$\\
\hspace*{2cm}whenever a message $[b,\ell,0]$, for some $b \neq a$, is decoded then\\
 \hspace*{3cm}send $[\ell,b,1]$
 \hfill$\diamond$
 
  \vspace{2ex}

 The procedure {\tt Spread}, used by the algorithm and executed by a node with label $\ell$, is described as follows.
 
 \vspace{2ex}
\noindent
{\bf Procedure} {\tt Spread}$(a_1,\dots,a_s)$

\vspace{2ex}
\noindent
send $[message]$\\
$i \leftarrow 1$\\
{\bf while} $i \leq s$ {\bf do}\\
\hspace*{1cm}send $[\ell,a_i,0]$\\
\hspace*{1cm}when the message $[a_i,\ell,1]$ is decoded then\\
\hspace*{2cm}$i \leftarrow i+1$
 \hfill$\diamond$

\begin{theorem}\label{aware}
Upon completion of Algorithm {\tt Neighborhood-aware} in an arbitrary $n$-node graph with $e$ edges, every node correctly decodes the source message.
The cost of the algorithm is $O(n\log M +e\log L)$.
\end{theorem}

\begin{proof}
In view of the disjointness property, all messages are correctly decoded by their addressees. Since the control messages $[a,b,0]$ and $[a,b,1]$ travel in a DFS fashion,
and each message $[a,b,0]$ is preceded by the source message $[m]$,  all nodes get the source message and decode it correctly. This proves the correctness of the algorithm. To estimate its cost, note that each node sends the source message $[m]$ once, and, for any pair of adjacent nodes $a$ and $b$, two control messages
among $[a,b,0]$, $[a,b,1]$, $[b,a,0]$, $[b,a,1]$ are sent. Since the source message $[m]$ consists of $O(\log M)$ beeps, and each control message consists of $O(\log L)$ beeps,
the total cost of the algorithm is $O(n\log M +e\log L)$.
\end{proof}

Before proving our lower bound on the cost of broadcasting algorithms  in neighborhood-aware networks, we prove the following two lemmas.

\begin{lemma}\label{log M}
Every broadcasting algorithm has cost $\Omega(\log M)$ in the two-node graph. 
\end{lemma}

\begin{proof}
Suppose that there exists a broadcasting algorithm that has cost at most $\frac{1}{2}\log M$ in the two node graph, where the node $u$ is the source,
and the node $v$ is the other node of the graph.
For any source message $m$, the adversary delivers all the beeps sent by $u$ in consecutive rounds.
Since there are fewer than $M$ different binary strings of length at most $\frac{1}{2}\log M$, for two different source messages, $m_1$ and $m_2$, the strings of beeps received by $u$ must be identical. Hence the message outputted by $v$ must be identical for  $m_1$ and $m_2$, and thus it must be incorrect for one of them.
\end{proof}

\begin{lemma}\label{log L}
Every broadcasting algorithm has cost $\Omega(\log\log L)$ in some cycle of size~4. 
\end{lemma}

\begin{proof}
Consider a  broadcasting algorithm $A$ working for all neighborhood-aware cycles of size 4.
Suppose that the cost of algorithm $A$ in all such cycles is at most $\frac{1}{2}\log \log L$.
Consider a cycle of size 4, and call its nodes $a,b,c,d$, in clockwise order. Suppose that node $a$ is the source. 
Let 0 be the label of node $a$, and let $L-1$ be the label of node $c$.
The adversary delivers all beeps possibly sent by node $c$, only after this node  outputs the source message. 
Hence, before the decision by node $c$, nodes $b$ and $d$ hear only beeps from the source $a$.
The adversary delivers all beeps sent by node $a$ in consecutive rounds. Since node $a$ can send at most $\frac{1}{2}\log \log L$ beeps,
the set $X$ of possible sequences of beeps heard by nodes $b$ and $d$ has size at most $\sqrt{\log L}$.
Let $N=\{0,1,\dots, \lfloor \frac{1}{2}\log \log L \rfloor\}$. Since each of the nodes $b$ and $d$ can send at most $\frac{1}{2}\log \log L$ beeps,
the number of beeps sent by each of these nodes must be an integer from the set $N$.
For any label $\ell \in \{1,\dots ,L-2\}$, let $\Phi _\ell : X \longrightarrow N$ be the function defined as follows: $\Phi _\ell (x)$ is the number of beeps sent by the node $b$ or $d$, if it has label $\ell$,
and if it obtained the sequence $x$ of beeps. There are $|N|^{|X|} <L-2$ such functions, for sufficiently large $L$. Hence there exist labels $\ell_1 \neq \ell_2$ from the set $ \{1,\dots ,L-2\}$, for which $ \Phi _{\ell_1}= \Phi _{\ell_2}$. Assign these labels to the nodes $b$ and $d$. In the obtained cycle $C$, nodes $b$ and $d$ send the same number of beeps, regardless of the sequence of beeps obtained from $a$. In particular, this will happen in two executions, $E_1$ and $E_2$, of algorithm $A$ on the cycle $C$, where execution $E_1$ corresponds to source message $m_1$, and
execution $E_2$ corresponds to source message $m_2$,  for $m_1 \neq m_2$. 

In both executions, the adversary delivers consecutive beeps from $b$ and from $d$ in the same rounds.
As a result, node $c$ hears only loud beeps: $k_1$ of them in execution $E_1$, and $k_2$ of them in execution $E_2$. 
Without loss of generality, suppose that $k_1 \leq k_2$.
The choice of the rounds of delivery of bits from $b$ and $d$ is as follows.
In execution $E_1$ these are consecutive rounds $r,r+1,\dots,r+k_1-1$, starting from some round $r$. Suppose that $s$ is the round in which  node $c$ correctly outputs message $m_1$.
Then, in execution $E_2$, the adversary delivers the first $k_1$ beeps from $b$ and $d$ in rounds $r,r+1,\dots,r+k_1-1$, and the remaining $k_2-k_1$ beeps in rounds
$t+1,\dots, t+k_2-k_1$, where $t=\max(s,,r+k_1-1)$. In round $s$, node $c$ has the same history in executions $E_1$ and $E_2$: it heard a loud beep in the same rounds, in both these executions.
Hence, in execution $E_2$, it incorrectly outputs the message  $m_1$.
\end{proof}

The following result gives a lower bound on the cost of any broadcasting algorithm in neighborhood-aware networks. 
\begin{theorem}
For arbitrarily large integers $n$, there exist $n$-node neighborhood-aware networks for which every broadcasting algorithm has cost 
$\Omega(n\log M +n\log\log L)$.
\end{theorem}

\begin{proof}
For any positive integer $k$, consider the graph $G_k$ defined as follows. Let $P_k$ be a simple path of length $k$, with extremities $a$ and $b$.
Consider pairwise disjoint copies $C_1,\dots,C_k$ of the cycle of size 4, whose all nodes are distinct from nodes of the path. Let $a_i, b_i, c_i, d_i$ be the nodes of the $i$th copy in clockwise order.
Join the node $a_1$ to the node $b$ by an edge, and for every $1 \leq i < k$, join the node $c_i$ to the node $a_{i+1}$ by an edge. The obtained graph has $n \in \Theta(k)$ nodes.
We now assign the labels to nodes of $G_k$ as follows. Nodes $b_i$ and $d_i$ in cycles 
$C_i$, for $i=1,\dots,k$, are assigned distinct labels by induction. For any $i$, we consider the set of all labels that were not used previously and find among them two labels $\ell_1 \neq \ell_2$ for which $ \Phi _{\ell_1}= \Phi _{\ell_2}$, where $ \Phi _{\ell}$, for any label $\ell$, was defined in the proof of Lemma  \ref{log L}. This can be done similarly as in the quoted proof, because the number of still available
labels is $\Theta(L)$. Finally, nodes of the path and all nodes $a_i$ and $c_i$ are assigned consecutive distinct labels among the remaining pool of labels.

Let the node $a$ be the source, and consider any broadcasting algorithm on graph $G_k$.
By Lemma \ref{log M}, each node of the path, other than $b$, has to transmit $\Omega(\log M)$ beeps, for otherwise the next node cannot get the message.
By Lemma  \ref{log L}, the total cost of the algorithm in each subgraph $C_i$, for $i<k$,  must be $\Omega(\log\log L)$, for otherwise the nodes of the next copy cannot get the message.
(Note that edges of the path $P_k$ and edges joining consecutive copies of the cycle, are bridges in $G_k$.) Hence the total cost of the algorithm is $\Omega(k\log M +k\log\log L)=\Omega(n\log M +n\log\log L)$.
\end{proof} 


\section{Full-knowledge networks}

In this section we consider broadcasting in networks whose nodes have the entire labeled map of the network, and know the identity of the source.
With this complete knowledge, all nodes can agree on the same spanning tree $T$ of the network, rooted at the source. (All trees rooted at the source can be canonically 
coded by binary strings, and the tree $T$ can be chosen as that with lexicographically smallest code.) Let $S$ be a DFS traversal of the tree $T$ in which children
of every node are explored in increasing order of their labels. The Eulerian tour of the tree $T$ corresponding to this traversal can be represented as a sequence
$(a_1,\dots, a_{2(n-1)})$ of length $2(n-1)$ of node labels with repetitions, where $a_i$ corresponds to the $i$th edge traversal in the tour, from the node with label $a_i$ to the node with label $a_{i+1}$. 

The only message circulating in the network is the message $[m]$, where $m$ is the source message, and $[m]$ is the encoding of this integer, described in Section 5.
The instruction send  $[m]$ is the procedure of sending beeps of the encoding $[m]$ in consecutive rounds.
Similarly as in  Algorithm {\tt Neighborhood-aware}, the disjointness property is satisfied, and hence each message can be correctly decoded by adjacent nodes.
The idea of the algorithm is the following. Every node knows to which terms of the sequence $(a_1,\dots, a_{2(n-1)})$ its label corresponds. It sends the message $[m]$ when the turn of such a
term of the sequence comes. (Many nodes send messages many times.) In order to know when this happens, the node computes how many previous messages it should get before from all adjacent nodes, and when this  number of messages is received, it proceeds with the execution of the send $[m]$ instruction corresponding to the given term of the sequence.

The algorithm is executed by a node with label $\ell$, when the source message is $m$. The pseudocode of the algorithm follows.

\vspace{2ex}
\noindent
{\bf Algorithm} {\tt Full-knowledge}

\vspace{2ex}
\noindent
{\bf if} the executing node is not the source  {\bf then}\\
\hspace*{1cm}when a message $[m]$ is decoded for the first time, then\\
\hspace*{2cm}output $message$ as the source message\\ 
identify all positions of label $\ell$ in the sequence $(a_1,\dots, a_{2(n-1)})$\\
let $i_1,\dots, i_r$ be these positions\\
let $x_1$ be the number of indices $1\leq j <a_1$, corresponding to labels $a_j$ of nodes adjacent to the node with label $\ell$\\
for $1< i \leq r$, let $x_i$ be the number of indices $a_{i-1}< j <a_i$, corresponding to labels $a_j$ of nodes adjacent to the node with label $\ell$\\ 
for $1< i \leq r$, let $y_i= \sum_{t=1}^i x_t$\\
{\bf for} $k=1$ {\bf to} $r$ {\bf do}\\
\hspace*{1cm}when a total of $y_k$ messages $[m]$ is received then send $[m]$
 \hfill$\diamond$

\begin{theorem}
Upon completion of Algorithm {\tt Full-knowledge} in an arbitrary $n$-node graph, every node correctly outputs the source message.
The cost of the algorithm is $O(n\log M)$.
\end{theorem}

\begin{proof}
The correctness of the algorithm follows from the fact that nodes send messages whenever their turn comes in the sequence $(a_1,\dots, a_{2(n-1)})$ 
that corresponds to an Eulerian tour of a spanning tree $T$, and from the disjointness property guaranteeing that the source message is always correctly decoded.
The total number of messages sent is $2(n-1)$. Since each message corresponds to $O(\log M)$ bits, the total cost of the algorithm is $O(n\log M)$.
\end{proof}

The following proposition shows that the cost of Algorithm {\tt Full-knowledge} is optimal in full-knowledge networks.

\begin{proposition}
For arbitrary integers $n \geq 2$ there exist $n$-node graphs for which the cost of any broadcasting algorithm is $\Omega(n\log M)$.
\end{proposition}

\begin{proof}
Consider the simple path $P_n$ with $n$ nodes, one of whose extremities is the source.
Note that Lemma \ref{log M} holds for full-knowledge networks as well.
By Lemma \ref{log M}, each node of the path, other than the last node, has to transmit $\Omega(\log M)$ beeps, for otherwise the next node cannot get the message
correctly. Hence the cost of any broadcasting algorithm is $\Omega(n\log M)$.
\end{proof}

\section{Conclusion}

We considered the cost of asynchronous broadcasting in networks with four different levels of knowledge: anonymous, ad-hoc, neighborhood-aware, and full-knowledge.
We proved that broadcasting in anonymous networks is impossible, and we showed upper and lower bounds on the cost of broadcasting for the other three levels of knowledge.
Our results show cost separations between all of them. While the bounds for full-knowledge networks are asymptotically tight, the other bounds are not, and designing optimal-cost broadcasting
algorithms for ad-hoc and for neighborhood-aware networks is a natural open problem.

 

\bibliographystyle{plain}

\begin{thebibliography}{12}
\bibitem{AABCHK}
Y. Afek, N. Alon, Z. Bar-Joseph, A. Cornejo, B. Haeupler, F. Kuhn, Beeping a maximal independent set.
Proc. 25th International Symposium on Distributed Computing (DISC 2011), LNCS 6950, 32-50.

\bibitem{AGPV}
B.  Awerbuch, O. Goldreich, D. Peleg, R. Vainish,  A trade-Off between information and communication in broadcast protocols, J. ACM 37 (1990), 238-256.

\bibitem{CFP}
T. Calamoneri, E. G. Fusco, A. Pelc, Impact of information on the complexity of asynchronous radio broadcasting, Proc. 12th International Conference on Principles of Distributed Systems (OPODIS 2008),  LNCS 5401, 311-330. 

\bibitem{CR}
B. S. Chlebus, M. A. Rokicki, Centralized asynchronous broadcast in radio networks. Theor. Comput. Sci. 383( (2007), 5-22.

\bibitem{CGR}
M. Chrobak, L. Gasieniec, W. Rytter, Fast broadcasting and gossiping in radio networks, J. Algorithms 43 (2002), 177-189.




\bibitem{CK}
A. Cornejo, F. Kuhn, Deploying wireless networks with beeps,
Proc. 24th International Symposium on Distributed Computing (DISC 2010), LNCS 6343, 148-162.

\bibitem{CD}
A. Czumaj, P. Davis, Communicating with beeps,  arxiv:1505.06107 [cs.DC] (2015)

\bibitem{EP}
S. Elouasbi, A. Pelc, Deterministic rendezvous with detection using beeps, Proc. 11th International Symposium on Algorithms and Experiments for Wireless Sensor Networks (ALGOSENSORS 2015).

\bibitem{FPRU}
U.~Feige, D.~Peleg, P.~Raghavan, E.~Upfal, Randomized broadcast in networks,
Random Structures and Algorithms 1 (1990), 447-460.

\bibitem{FSW}
K.-T. Forster, J. Seidel, R. Wattenhofer,
Deterministic leader election in multi-hop beeping networks,
Proc.28th International Symposium on Distributed Computing (DISC 2014), LNCS 8784, 212-226.


\bibitem{FrLa}
P.~Fraigniaud, E.~Lazard, Methods and problems of communication in usual 
networks, Disc. Appl. Math. 53 (1994), 79-133.

\bibitem{GH}
M. Ghaffari, B. Haeupler,
Near optimal leader election in multi-hop radio networks,
Proc. 24th Annual ACM-SIAM Symposium on Discrete Algorithms (SODA 2013),  748-766.

\bibitem{GN}
S. Gilbert, C. Newport, The computational power of beeps, Proc. 29th International Symposium on Distributed Computing (DISC 2015), 31-46.

\bibitem{HHL}
S.M.~Hedetniemi, S.T.~Hedetniemi, A.L.~Liestman, A survey
of gossiping and broadcasting in communication networks, Networks 18 (1988),
319-349.

\bibitem{HMP}
K. Hounkanli, A. Miller,  A. Pelc,
Global Synchronization and Consensus Using Beeps in a Fault-Prone MAC. CoRR abs/1508.06583 (2015)


\bibitem{Ko}
D. Kowalski, Algorithmic Foundations of Communication in Radio Networks, Morgan \& Claypool Publishers, 2011.

\bibitem{KP}
D. Kowalski, A. Pelc, Time of deterministic broadcasting in radio networks with local knowledge, SIAM Journal on Computing 33 (2004), 870-891.

\bibitem{KM}
E. Kushilevitz, Y. Mansour,  An Omega(D log (N/D)) lower bound for broadcast in radio networks. SIAM Journal on Computing  27 (1998), 702-712.

\bibitem{MRZ}
Y. M{\'{e}}tivier, J. M. Robson, A. Zemmari, On distributed computing with beeps, CoRR abs/1507.02721 (2015).
 

\bibitem{Pe}
A. Pelc, Fault-tolerant broadcasting and gossiping in communication networks, Networks 28 (1996), 143-156. 

\bibitem{Pe2}
A. Pelc, Activating anonymous ad hoc radio networks, Distributed Computing 19 (2007), 361-371. 

\bibitem{SCH}
P. J. Slater, E. J. Cockayne, S. T. Hedetniemi,  Information dissemination in trees, SIAM Journal on Computing 10 (1981), 692-701.

\bibitem{YJYLC}
J. Yu, L. Jia, D. Yu,  G. Li, X. Cheng, Minimum connected dominating set construction in wireless networks
               under the beeping model, Proc. IEEE Conference on Computer Communications, (INFOCOM 2015), 972-980.






\end{thebibliography}


\end{document}